\pdfoutput=1
\documentclass[letterpaper, 10 pt, conference]{ieeeconf}  

\IEEEoverridecommandlockouts                              

\usepackage{graphicx}
\usepackage{amsmath} 
\usepackage{amssymb}
\usepackage{mathrsfs}
\usepackage{dsfont}

\usepackage{algorithm,algorithmic}
\newtheorem{thm}{\bf Theorem}[section]
\newtheorem{cor}[thm]{\bf Corollary}
\newtheorem{lem}[thm]{\bf Lemma}
\newtheorem{rem}[thm]{\bf Remark}
\newtheorem{ass}[thm]{\bf Assumption}
\newtheorem{prop}[thm]{\bf Proposition}

\newtheorem{deff}[thm]{\bf Definition}
\newtheorem{eg}[thm]{\bf Example}

\newcommand{\eps}{\varepsilon}

\newcommand{\ls}{\mathcal{L}}

\newcommand{\R}{\mathbb{R}}
\newcommand{\N}{\mathbb{N}}
\newcommand{\ppp}{\mathbf{P}}

\newcommand{\eee}{\mathbf{E}}

\newcommand{\tx}{\tilde{x}}

\newcommand{\om}{\Omega}

\newcommand{\xx}{\mathcal{X}}
\newcommand{\uu}{\mathcal{U}}
\newcommand{\uf}{\mathfrak{u}}
\newcommand{\psp}{\mathfrak{P}}
\newcommand{\ck}{\mathfrak{C}}

\newcommand{\pp}{\mathcal{P}}
\newcommand{\ff}{\mathcal{F}}

\newcommand{\trans}{\mathcal{T}}

\newcommand{\set}[1]{\left\{#1\right\}}
\newcommand{\ws}[1]{\left\|#1\right\|_{\operatorname{W}}}

\newcommand{\tv}[1]{\left\|#1\right\|_{\operatorname{TV}}}

\newcommand{\Qs}{\mathcal{Q}}
\newcommand{\ball}{\overline{\mathbb{B}}}
\newcommand{\oball}{\mathbb{B}}
\newcommand{\ballr}{\overline{\mathcal{B}}}
\newcommand{\oballr}{\mathcal{B}}

\newcommand{\LL}{\mathcal{L}}
\newcommand{\LLk}{\mathfrak{L}}

\newcommand{\cj}{\wedge}
\newcommand{\br}{\mathscr{B}}

\newcommand{\ep}{\vartheta}

\newcommand{\inte}{\operatorname{Int}}
\usepackage{color}

\newcommand{\ra}{\rightarrow}

\newcommand{\sub}{\subseteq}

\newcommand{\wh}{\widehat}
\newcommand{\ob}{\mathbf{\Omega}}
\newcommand{\w}{\varpi}

\newcommand{\tta}{\Theta}
\newcommand{\ts}{\mathcal{T}}

\newcommand{\ap}{\operatorname{AP}}

\newcommand{\fscr}{\mathscr{F}}
\newcommand{\Ws}{\mathcal{W}}

\newcommand{\nuo}{\mathcal{\nu}_0}

\newcommand{\wra}{\rightharpoonup}
\newcommand{\G}{\mathfrak{G}}
\newcommand{\bw}{\mathcal{B}_W}
\newcommand{\pim}{\mathcal{P}}

\newcommand{\T}{\mathbb{T}}
\newcommand{\ttil}{\widetilde{\mathcal{T}}}
\newcommand{\transs}[1]{\left[\![#1\right]\!]}
\newcommand{\csig}{\Sigma}
\newcommand{\I}{\mathbb{I}}
\newcommand{\Q}{\mathcal{Q}}

\newcommand{\tmou}{\widetilde{\transs{\mathcal{T}^u}}}
\newcommand{\tnu}{\tilde{\nu}}
 \newcommand{\tdt}{\widetilde{\mathcal{T}}}
 \newcommand{\fl}[1]{\lfloor{#1}\rfloor}
 
\newcommand{\rff}{{\operatorname{ref}}}
\newcommand{\act}{{\textit{Act}}}
\newcommand{\MU}{\mathbb{XU}}
\newcommand{\IU}{\mathbb{IA}}
\newcommand{\af}{\mathfrak{a}}

\newcommand{\wb}{\mathbf{w}}

\usepackage[all]{xy}
\usepackage{pifont}

\newcommand{\undert}{\check{\Theta}}
\newcommand{\overt}{\hat{\Theta}}

\title{\LARGE \bf Robustly Complete Finite-State
Abstractions for 
Control Synthesis of Stochastic Systems}

\author{Yiming Meng and Jun Liu
\thanks{Preprint submitted to IEEE Open Journal of Control Systems.}
\thanks{Yiming Meng and Jun Liu are  with Department of Applied Mathematics,
        University of Waterloo, Ontario, Canada, 
        {\tt\small \{yiming.meng, j.liu\}@uwaterloo.ca}}%
}

\begin{document}
\maketitle
\thispagestyle{plain}
\pagestyle{plain}

\begin{abstract}

The essential step of abstraction-based control synthesis for nonlinear systems to satisfy a given specification is to obtain a finite-state abstraction of the original systems. The complexity of the abstraction is usually the dominating factor that determines the efficiency of the algorithm. For the control synthesis of discrete-time nonlinear stochastic systems modelled by nonlinear stochastic difference equations, recent literature has demonstrated the soundness of abstractions in preserving robust probabilistic satisfaction of $\omega$-regular linear-time properties. However, unnecessary transitions exist within the abstractions, which are difficult to quantify, and the completeness of abstraction-based  control synthesis in the stochastic setting  remains an open theoretical question. In this paper, we address this fundamental question from the topological view of metrizable space of probability
measures, and propose  constructive finite-state abstractions for control synthesis of probabilistic linear temporal specifications. Such abstractions are both sound and approximately complete. That is, given a concrete discrete-time stochastic system and an arbitrarily small $\mathcal{L}^1$-perturbation of this system, there exists a family of finite-state controlled Markov chains that both abstracts the concrete system and is abstracted by the slightly perturbed system. In other words, given  an arbitrarily small prescribed precision, an abstraction  always exists to decide whether a  control strategy exists for the concrete system to satisfy the probabilistic specification.  
\end{abstract}


\maketitle

\section{INTRODUCTION}

Abstraction-based formal synthesis  relies on  obtaining a finite-state abstraction (or symbolic model) of the original, and possibly nonlinear systems.  Computational methods, such as graph-based model checking and automaton-guided controller synthesis, are then developed based on the abstraction to verify the system or synthesize controllers with respect to a temporal logic specification \cite{baier2008principles,belta2017formal,li2023formal}. Abstractions enable autonomous decision making of physical systems to achieve more complex tasks, and received significant success in the past decade \cite{belta2017formal,liu2016finite, nilsson2017augmented, ozay2013computing,li2023formal}. Regardless of heavy state-space
discretization and complicated abstraction analysis, formal methods compute with guarantees a set of initial states from which a
controller exists to realize the given specification \cite{belta2017formal, liu2017robust, li2020robustly,li2023formal}.  

Heuristically, abstractions use a finite-state automaton to solve the corresponding search
problem  at a cost of potentially including non-deterministic  transition in the automation. 
  For non-stochastic control systems, both sound and approximately complete abstractions exist \cite{kloetzer2008fully,tabuada2006linear, liu2017robust,li2020robustly,liu2021closing,li2023formal}. This is in the sense that, the abstractions can not only include a sufficient number of  transitions  to design
provably correct controllers, but also quantify the level of over-approximation allowed for a specified precision. Therefore, the  completeness analysis theoretically removes the doubts of finding an abstraction-based approach once a robust control strategy of a certain degree is supposed to exist with respect to a given  specification,  which makes any computational
attempts not meaningless.

There is a recent surge of interest in studying formal methods for stochastic systems in verification and control synthesis of probabilistic specifications. Formal analysis of stochastic abstractions relies on different mathematical techniques. We review some crucial results from the literature that are pertinent to the work presented in this paper.

\subsection{Related Work}
Probabilistic model checkers have been developed for discrete-time discrete-state fully observed  Markov decision processes (MDP) and partially observed  Markov decision processes (POMDP) \cite{kwiatkowska2002prism,parker2013verification, chatterjee2016decidable,hensel2021probabilistic, kwiatkowska2022probabilistic}, and have gained success in applications of control synthesis with  probabilistic temporal logics \cite{ding2011mdp, ding2011ltl, lacerda2014optimal, wells2020ltlf}. 

For continuous-state space, a major strategy is to approximate the transition kernels by some reference (control) stochastic matrices,  known as finite-model approximations, to solve optimal control problem or control synthesis with respect to probabilistic temporal logics \cite{saldi2017finite, norman2017verification}. Probabilistic reachability and safety related control synthesis can be resolved by relating the satisfaction probability 
to the corresponding
value functions 
\cite{ramponi2010connections,soudjani2011adaptive,summers2010verification,abate2008probabilistic,abate2011quantitative}. By necessarily imposing stability-like conditions, the problem can be reduced to
solving the 
 characterized dynamic programming problem  using computable bounded-horizon counterparts \cite{tkachev2011infinite,tkachev2012regularization, tkachev2013formula, tkachev2014characterization,tkachev2017quantitative}. 

For fully observed systems, other than the approximation schemes, a formal abstraction  for stochastic systems provides an inclusion of all possible  approximate transitions of the   labelled  processes, which  eventually will  preserve the probability of satisfaction in a proper sense.   Bounded-Parameter Markov Decision Processes (BMDP) can naturally serve this purpose \cite{givan2000bounded, lahijanian2015formal}.  A BMDP   contains  a
family of finite-state MDPs with uncertain transitions given each action, and   provide the upper and lower quasi-stochastic matrices as abstractions
for the continuous-state controlled Markov systems. The authors of \cite{lahijanian2015formal}  developed algorithms based on \cite{givan2000bounded,wu2008reachability}  to obtain the upper/lower bound of the satisfaction probability of fundamental formulas of probabilistic computation tree logic.  The work in \cite{cauchi2019efficiency} formulated BMDP abstraction for  bounded-linear temporal logic  specifications. The most recent works \cite{dutreix2020verification, dutreix2022abstraction} for the first time developed a specification-guided refinement strategy on the partition of the state space and presented a synthesis procedure for finite-mode discrete-time stochastic systems against any $\omega$-regular specifications. All the above mentioned abstraction techniques appear to be sound but
not complete, apart from  \cite{abate2008markov} under strong assumptions.

The recent research \cite{meng2022robustly} proposed a notion of completeness for stochastic abstractions in verification of probabilistic $\omega$-regular properties. That is,  
given a concrete discrete-time continuous-state Markov process $X$, 
and an arbitrarily small $\ls^1$-bounded perturbation of this system, there always
exists an Interval Markov Chains (IMC) abstraction whose interval of satisfaction probability contains that of $X$, and meanwhile is contained  by that of the slightly perturbed
system. Instead of imposing the mix-monotone conditions \cite{dutreix2020specification} and  the strong stability (ergodicity) assumptions  \cite{abate2008markov} of the stochastic systems, the analysis in \cite{meng2022robustly} is based on the topology of metrizable space of probability measures with only mild conditions.   
This methodology proves to be more effective than simply discussing the value of probabilities and enables us to demonstrate the approximate completeness of abstraction-based stochastic control synthesis. 

\subsection{Contributions}

In this paper, we establish theoretical results on abstraction-based control of discrete-time nonlinear Markov systems, building on recent work \cite{meng2022robustly} on formal verification. In brief:
\begin{itemize}
\item 
We define abstractions based on the topology of the metrizable space of probability measures and propose the concept of robust completeness for controlled Markov systems.

\item While it is often believed to be true in literature that abstraction-based stochastic control synthesis is sound (e.g., Fact 1 of \cite{dutreix2020specification, dutreix2022abstraction}, as well as similar statements in \cite{lahijanian2015formal, laurenti2020formal, delimpaltadakis2021abstracting}), we provide the first formal proof of its soundness, to the best knowledge of the authors.

\item We prove that robustly complete abstractions of fully observed controlled Markov  systems (even with additional uncertainties) exist under a mild assumption, which demonstrates the decidability of robust realization of probabilistic $\omega$-regular temporal logic formulas.

\item We improve upon the analysis in \cite{meng2022robustly} by providing a set of tighter inequalities that avoid unnecessarily refined partitions of the state space to guarantee the prescribed precision. 

\item We discuss the applicability of formal abstractions to partially observed controlled stochastic systems. 

\end{itemize}

The rest of the paper is organized as follows. Section \ref{sec: pre} presents some preliminaries on probability spaces and controlled Markov systems. Section \ref{sec: bmdp} presents the soundness of abstractions in verifying $\omega$-regular linear-time properties for fully observed discrete-time controlled Markov systems.  Section \ref{sec: complete} presents the constructive robust
abstractions with soundness and  approximate completeness guarantees. We discuss the applicability of the proposed method for partially observed discrete-time controlled Markov systems in Section \ref{sec: noisy}.  The paper is concluded in
Section \ref{sec: conclusion}.

\subsection{Conventions for Notation}
We denote by $\small{\prod}$ the product of ordinary sets, spaces, or function values. Denote by $\otimes$ the product of collections of sets, or sigma algebras, or measures. The $n$-times repeated product of any kind is denoted by $(\cdot)^n$ for simplification. Denote by $\pi_j:\prod_{i=0}^\infty(\cdot)_i\rightarrow(\cdot)_j$ the projection to the $j^{\text{th}}$ component.  We denote the Borel $\sigma$-algebra of a set by $\mathscr{B}(\cdot)$ and the space of all probability measures on $\br(\cdot)$ by $\psp(\cdot)$. 

For a set $A\subseteq\R^n$, $\overline{A}$ denotes its closure, $\inte(A)$ denotes its interior, and $\partial A$ denotes its boundary.  
For two sets $A,B\sub\R^n$, the set difference is defined by $A\setminus B=\set{x:\,x\in A,\,x\not\in B}$.

Let $|\cdot|$ denote the inifinity norm in $\R^n$ and let  $ \oball:=\{x\in \R^n: |x|<1\}$. Given a probability space $(\ob, \ff,\ppp)$, we denote by $\|\cdot\|_1:=\eee|\cdot|$ the $\ls^1$-norm for $\R^n$-valued random variables, and let $ \oballr:=\{X: \R^n\text{-valued random variable with}\;\|X\|_1<1\}$. 

Given a matrix $M$, we denote by $M_i$ its $i^{\text{th}}$ row and by $M_{ij}$ its entry at the $i^{\text{th}}$ row and $j^{\text{th}}$ column.

\section{PRELIMINARIES}
\label{sec: pre}

We consider $\mathbb{N}=\{0,1,\cdots\}$ as the discrete time index set, and a general Polish (complete and separable metric) space $\xx$ as the state space. Let $\uu\sub\R^p$ be a compact space of control inputs. We introduce some standard concepts for fully observed controlled Markov processes.

\subsection{Canonical Setup for Discrete-Time Controlled Markov Processes}\label{sec: canon}
The canonical setup for discrete-time controlled processes is provided in \cite{gihman2012controlled}. In brief, without loss of generality, we assume that a stochastic process $X:=\{X_t\}_{t\in\N}$ and a  process of control values $\uf:=\{\uf_t\}_{t\in\N}$ are defined on some (unknown) probability space where the noise is generated. Given any measurable process $\uf$, the probability law of the joint process $(X,\uf):=\{(X_t,\uf_t)\}_{t\in\N}$ can be determined on the canonical space   $((\xx\times\uu)^\infty, \ff, \ppp)$,  where 
$$\ff:=\sigma\{(X_t,\uf_t)\in(\Gamma,\ck),\; (\Gamma, \ck)\in\mathscr{B}(\xx)\otimes\mathscr{B}(\uu),\;\;t\in\mathbb{N}\}.$$
We also denote $X^\uf$ by the controlled process if we emphasize on the state-space marginal of $(X,\uf)$.

We  consider $(X,\uf)$ to be obtained from Markov models, whose transition probabilities, unlike control-free systems, have an extra dependence of the current control input, i.e., 
\begin{equation}
    \tta_t^u(x,\Gamma)=\ppp[X_{t+1}\in\Gamma\;|\;X_t=x,\uf_t=u].
\end{equation}

Now we suppose that
$\uf_t$ is provided according to some  rule at each instant of time $t\in\N$. It is natural
to suppose that the selection of a control at time $t$ is based on the history 
$X_{[0,t]}$ and 
$\uf_{[0,t-1]}$, where  
\begin{equation}\label{E: history_notation}
    X_{[0,t]}:=\{X_s\}_{s\in[0,t]}\;\text{and}\;\;\uf_{[0,t]}:=\{\uf_s\}_{s\in[0,t]}.
\end{equation}
For each fixed $t>0$, let $\kappa_t(\cdot\;|\;\cdot)$ be such that, for any $\ck\in\mathscr{B}(\uu)$,
\begin{equation}\label{E: policy}
    \kappa_t(\ck\;|\;X_{[0,t]}; \uf_{[0,t-1]})=\ppp[\uf_t\in\ck\;|\;X_{[0,t]}; \uf_{[0,t-1]}].  
\end{equation}
A control policy is defined as follows. 
\begin{deff}\label{def: policy}
An admissible control policy is the sequence $$\kappa=\{\kappa_t,\;t\in\N\},$$
where, for each $t\in\N$, $\kappa_t$ is given  in the form of \eqref{E: policy}. 
\end{deff}
If $\uf$ is generated based on a control policy $\kappa$, we replace the notation $X^\uf$ by $X^\kappa$. 
\begin{ass}
We assume that $\uf$ is deterministic in \textit{a priori} or generated by deterministic control policies.
\end{ass}

\subsection{Controlled Markov Systems}
We are interested in controlled 
Markov processes with discrete labels of states, which is done by assigning abstract labeling functions over a finite set of atomic propositions. Now we consider an abstract family of labelled controlled Markov  processes as follows.

\begin{deff}[Controlled Markov system]\label{def: MDP}
A  controlled Markov  system is a tuple $\MU=(\xx, \uu, \{\tta\}, \ap, L)  $, where
\begin{itemize}
\item $\xx=\Ws\cup\Delta$, where $\Ws$ is a bounded working space, $\Delta:=\Ws^c$ represents all the out-of-domain states;
\item $\uu$ is the set of actions;
\item $\{\tta\}:=\{[\![\tta^u]\!]\}_{u\in\uu}$ contains all collections of control-dependent transition probabilities:  for  every $t$, given a realization $u\in\uu$ of the signal $\uf_t$,  the transition $\tta_t^u$ is chosen from the collection $\transs{\tta^u}$ accordingly;
\item $\ap$ is the finite set of atomic propositions;
\item$L:\xx\rightarrow 2^{\ap}$ is the (Borel-measurable) labelling function, i.e. for every $A\in\mathscr{B}(2^{\ap})$, $L^{-1}(A)\in\ff$.

\end{itemize}
\end{deff}

Note that for every given $\uf$ and initial condition $X_0=x_0$ (resp. initial distribution $\nu_0\in\psp(\xx)$),  we can generate a process $X^\uf\in\MU^\uf$, whose probability law is denoted by $\ppp^{x_0,\uf}_{X}$ (resp. $\ppp^{\nu_0,\uf}_{X}$), and $\MU^\uf$ denotes all the processes that are generated by $\{\tta\}$ given $\uf$. The collection of all the probability laws of such controlled processes is denoted by $\{\ppp^{x_0,\uf}_{X}\}_{X^\uf\in\MU^\uf}$ (resp. $\{\ppp^{\nu_0,\uf}_{X}\}_{X^\uf\in\MU^\uf}$). We denote by $\{\ppp_{n}^{x_0, \uf}\}_{n=0}^\infty$ (resp. $\{\ppp_{n}^{\nuo, \uf}\}_{n=0}^\infty$) a sequence of $\{\ppp^{x_0,\uf}_{X}\}_{X^\uf\in\MU^\uf}$ (resp. $\{\ppp^{\nu_0,\uf}_{X}\}_{X^\uf\in\MU^\uf}$). We simply use $\ppp_X^\uf$ (resp. $\{\ppp_X^\uf\}_{X^\uf\in\MU^\uf}$) if we do not emphasize the initial condition (resp. distribution).

If $\uf$ is known to be generated according to some deterministic control policy $\kappa$, the previously mentioned notations are changed correspondingly by replacing the superscripts $\uf$ by $\kappa$. If $\kappa$ is not emphasized in the context, we use the superscripts $\uf$ to indicate the general controlled quantities.

\begin{deff}[Clarification of Notation]\label{def: clarification}
In the specific context of discrete state space $\xx$, given a controlled Markov process $X^\uf$ on $\xx^\infty$, we  use the notation $(\om,\fscr,\pp_X^\uf)$ for the discrete canonical spaces of some discrete-state controlled process. We would like to still use the notation $(\ob,\ff,\ppp_X^\uf)$ if the continuity of $\xx$ is not clear or not emphasized.
\end{deff}

For a path of controlled state $\w:=\w_0\w_1\w_2\cdots\in\xx^\infty$, define by 
$L_\w:=L(\w_0)L(\w_1)L(\w_2)\cdots$ its trace.  
The space of infinite words is denoted by $$(2^{\ap})^\omega=\{A_0A_1A_2\cdots:A_i\in 2^{\ap},\;i=0,1,2\cdots\}.$$ 
A linear-time (LT) property is a subset of $(2^{\ap})^\omega$. We are only interested in LT properties $\Psi$ such that $\Psi\in\mathscr{B}((2^{\ap})^\omega)$, i.e., those are Borel-measurable\footnote{By \cite{tkachev2017quantitative} and \cite[Proposition 2.3]{vardi1985automatic}, any $\omega$-regular language of labelled (controlled) Markov processes is measurable. The proof relies on the properties of the canonical space as well as the connection with B\"{u}chi automation.}. 

 To connect with $\omega$-regular specifications, we introduce the semantics of path satisfaction as well as probabilistic satisfaction as follows.
 
 \begin{deff}
Suppose $\Psi$ is a formula of our interest. For a given labelled controlled Markov process $X^\uf$ from $\MU^\uf$  with initial distribution $\nuo$, we formulate the canonical space  $(\ob,\ff,\ppp_X^{\nuo, \uf})$. For a controlled path $\w\in\xx^\infty$, we define the path satisfaction as
$$\w\vDash \Psi\Longleftrightarrow L_\w\vDash \Psi. $$
We denote by $\{ X^\uf\vDash\Psi\}:=\{\w: \;\w\vDash\Psi\}\in\ff$ 
 the events of path satisfaction. Given a specified probability $\rho\in[0,1]$, we define the probabilistic satisfaction of $\Psi$ as
$$X^\uf\vDash \ppp^{\nuo, \uf}_{\bowtie\rho}[\Psi]\Longleftrightarrow\ppp_X^{\nuo, \uf}\{X^\uf\vDash\Psi\}\bowtie \rho, $$
where $\bowtie\in\{\leq, <,\geq,>\}$.
\end{deff}

\subsection{The Concrete Controlled Markov Systems}
We focus on controlled Markov processes determined by the following fully-observed Markov system
\begin{equation}\label{E: sys_abstraction_control}
    X_{t+1}=f(X_t, \uf_t)+b(X_t)\wb_t+ \ep\xi_t,
\end{equation}
the sample state  $X_t^\uf(\w)\in\xx\subseteq\R^n$ for all $t\in\mathbb{N}$ given a signal process $\uf$, the stochastic inputs $\{\wb_t\}_{t\in\mathbb{N}}$ are i.i.d.  Gaussian random variables with covariance $ I_{k\times k}$ without loss of generality. Mappings $f:\R^n\times\R^p\ra\R^n $ is locally Lipschitz continuous in both arguments, and $b:\R^n\rightarrow\R^{n\times k}$ is locally Lipschitz continuous.  The memoryless perturbation $\xi_t\in\ballr$ are independent random variables with intensity $\ep\geq 0$ 
and unknown distributions.  We can translate \eqref{E: sys_abstraction_control} into the form of a  controlled Markov system
\begin{equation}\label{E: sysr_control}
    \MU=(\xx,\uu,  \{\mathcal{T}\}, 
    \ap,L_\MU),
\end{equation}
where $\{\ts\}:=\{\transs{\ts^u}\}_{u\in\uu}$ is defined in the same way as the $\{\tta\}$ in Definition \ref{def: MDP}. We use notation $\ts$ instead of $\tta$ to indicate the continuity of the transition probability in $x\in\xx$.
\begin{rem}
For $\ep\neq 0$, due to the $\ls^1$-bounded uncertainties, 
\eqref{E: sys_abstraction_control} defines a family $\MU$ of controlled Markov processes.  As to simulate the probability laws at the observation times,  the above system can be regarded as a discrete-time numerical scheme of controlled stochastic differential equations (SDEs) driven by Brownian motions, which demonstrates practical meanings in physical sciences and finance. We will show in the Section \ref{sec: complete} that any uniformly integrable noise with known distribution can play the role of $\{\wb_t\}_{t\in\mathbb{N}}$ in the completeness analysis. The real noise with bounded supports that are considered  in \cite{dutreix2020specification} is a special type. Gaussian variables in \eqref{E: sys_abstraction_control}  do not lose any generalities in view of $\ls^1$ properties and are in favor of our formal analysis. 

In addition, compared to $f$ being mixed-monotone and $b$  being constant in \cite{dutreix2020specification, dutreix2022abstraction}, the choice of $f$ and $b$ in this paper fits more general dynamics in applications.
\end{rem}

For real-world applications, we only care about the behaviors in the bounded working space $\Ws$. It is desired to trap the sample paths at the out-of-domain states once $\Delta$ they reach $\Delta$. By defining stopping time $\tau=\tau(\uf):=\inf\{t\in\mathbb{N}: X^\uf\notin\Ws\}$ for each $X^\uf$, it is equivalent to study the probability law of the corresponding stopped process $\{X^\uf_{t\cj\tau}\}_{t\in\N}$ for any initial condition (or distribution), which coincides with $\ppp_X^{\uf}$ on $\Ws$. In view of the corresponding transitions probability, for each realization of control input $u$ and for all $x\in\xx\setminus\Ws$, the transition probability should satisfy $\mathcal{T}^u(x,\Gamma)=0$ for all $\Gamma$ such that $\Gamma\cap \Ws\neq \emptyset$. 

\begin{rem}
It is worth noting that, in the numerical examples in \cite{dutreix2020verification, dutreix2022abstraction}, the  peudo-Gaussian noise with a bounded support is obtained by  normalizing real Gaussian distribution on $\Ws$ by the probability $\mathcal{N}(0, 1)(\Ws)$, which significantly distorts the shape of Gaussian density within $\inte(\Ws)$.  The treatment of out-of-domain transitions in this paper should preserve the density of $\wb_t$ and hence that of $X_t$ for each $t$ on $\inte(\Ws)$.  The densities can be recovered to the true densities given the stopping time $\tau$ not being triggered. 
\end{rem}

\begin{deff}[Clarification of Notations]\label{def: clarification2}
To avoid any complexity, we use the same notation $X^{\uf}$ and $\ppp_X^{\uf}$ to denote the stopped processes and the associated laws.
\end{deff}

\begin{ass}\label{ass: as1_abstraction}
We assume that $\textbf{in}\in L(x)$ for any $x\notin\Delta$ and $\textbf{in}\notin L(\Delta)$.
We can also include `always $(\textbf{in})$' in the specifications to observe sample paths for `inside-domain' behaviors, which is equivalent to verifying $\{\tau=\infty\}$.
\end{ass}

\subsection{Weak Topology}
Since our purpose is to investigate the relation  between continuous-state and finite-state controlled Markov systems and then demonstrate probabilistic regularities, it is natural to work on the dual space of the state space, i.e., we consider the set of possibly uncertain measures within the topological space of probability measures.  

Consider any separable and complete state space (Polish space) $\xx$.
 The following concepts on the space of probability measures $\psp(\xx)$\footnote{$\psp(\xx)$ is always metrisable given $\xx$ is a Polish space. \cite{rogers2000diffusions}} are frequently used later.  Note that, `if a space is metrisable, the topology is determined by convergences of sequences, which explains
we sometimes only define the concept of convergence,
without explicitly mention the topology.' \cite{hairer2020notes}  

\begin{deff}[Weak convergence\footnote{Technically, this should be weak* convergence as probability measures are  linear bounded functionals  of bounded continuous functions. We use `weak convergence' due to no conflict of concept in this paper. }]\label{def: weak_conv}
A sequence $\{\mu_n\}_{n=0}^\infty\subseteq\psp(\xx)$ is said to converge weakly to a probability measure $\mu$, denoted by $\mu_n\wra\mu$,  if 
\begin{equation}
    \int_\xx h(x)\mu_n(dx)\rightarrow \int_\xx h(x)\mu(dx),\;\;\forall h\in C_b(\xx).
\end{equation}
We frequently use the following 
alternative condition \cite[Proposition 2.2]{da2014stochastic}:
\begin{equation}
    \mu_n(A)\rightarrow\mu(A),\;\;\forall A\in\mathscr{B}(\xx) \;\text{s.t.} \;\mu(\partial A)=0.
\end{equation}

Correspondingly, the weak equivalence of any two measures $\mu$ and $\nu$ on $\xx$ is  such that 
\begin{equation}
    \int_\xx h(x)\mu(dx)=\int_\xx h(x)\nu(dx),\;\;\forall h\in C_b(\xx).
\end{equation}
\end{deff}

\begin{eg}\label{eg: strong_conv}
It is interesting to note that $x_n\rightarrow x$ 
in $\xx$ does not imply the strong convergence 
of the associated Dirac measures. However, we do have $\delta_{1/n}\wra\delta_0$.  A classical counterexample is to let $x_n=1/n$ and $x=0$, and we do not have $\lim_{n\rightarrow\infty}\delta_{1/n}= \delta_0$ in the strong sense since, i.e., 
$0=\lim_{n\rightarrow\infty}\delta_{1/n}(\{0\})\neq \delta_0(\{0\})=1. $ We kindly refer readers to \cite{rogers2000diffusions, sagar2021compactness} and \cite[Remark 3]{meng2022complete} for more details on the weak topology.
\end{eg}

\begin{deff}[Tightness of set of measures]\label{def: tight}
Let $\xx$ be any topological state space  and $M\subseteq\psp(\xx)$ be a set of probability measures on $\xx$. 
We say that $M$ is  tight if, for every $\eps>0$ there exists
a compact set $K\subseteq \xx$ such that $\mu(K)\geq 1-\eps$ for every $\mu\in M$.
\end{deff}
The following theorem provides an alternative criterion for verifying the compactness of family of measures w.r.t. the corresponding metric space using tightness. Note that, on a compact metric space $\xx$, every family of probability measures is tight.
\begin{thm}[Prokhorov]\label{thm: proh}
Let $\xx$ be a complete separable metric space. A family $\Lambda\subseteq\psp(\xx)$ 
is relatively compact 
if and only if it is tight. Consequently, for each sequence $\{\mu_n\}$ of tight $\Lambda$, there exists a $\mu\in\bar{\Lambda}$ and a subsequence $\{\mu_{n_k}\}$ such that $\mu_{n_k}\wra\mu$.

\end{thm}

\subsection{Robust Abstractions}
We define a notion of abstraction between continuous-state and finite-state controlled Markov systems via state-level relations and measure-level relations. 
\begin{deff}
A (binary) relation $\gamma$ from $A$ to $B$ is a subset of $A\times B$ satisfying (i) for each $a\in A$, $\gamma(a):=\{b\in B: (a,b)\in \gamma\}$; (ii) for each $b\in B$, $\gamma^{-1}(b):=\{a\in A: (a,b)\in \gamma\}$; (iii) for $A'\subseteq A$, $\gamma(A')=\cup_{a\in A'}\gamma(a)$; (iv) and for $B'\subseteq B$, $\gamma^{-1}(B')=\cup_{b\in B'}\gamma^{-1}(b)$.
\end{deff}

\begin{deff}\label{def: abs_control}
Given a continuous-state controlled Markov system
$$\MU=(\xx, \uu,  \{\mathcal{T}\},\ap,L_\MU) $$
with a compact $\uu\in\R^p$, 
and a finite-state Markov system
$$\IU=(\Q, \act, \{\tta\}, \ap, L_\IU),$$
where $\Q=(q_1,\cdots,q_N)^T$, $\act=\{a_1,\cdots,a_M\}$, and $\{\tta\}:=\{[\![\tta^a]\!]\}_{a\in\act}$ contains all collections of $n\times n$ stochastic matrices that are also dependent on $a$. 

We say that $\IU$ abstracts $\MU$, and write $\MU\preceq _{\csig_\alpha} \IU$, if there exist
\begin{itemize}
    \item[(1)] a  state-level relation $\alpha\subseteq \xx\times \Q$ from $\MU$ to $\IU$ such that,  for all $x\in \xx$, there exists $q\in \Q$ such that $(x,q)\in\alpha$ ($\alpha(x)\neq \emptyset$) and $L_\IU(q)=L_\MU(x)$;
    \item[(2)] a measure-level relation
$\csig_\alpha\subseteq \psp(\xx)\times\psp(Q)$ from $\MU$ to $\IU$ such that, for all $i\in \{1,2,\cdots, N\}$ and $a\in\act$, there exists  $u\in\uu$ such that for any $\mathcal{T}^u\in\transs{\mathcal{T}^u}$ and all $x\in\alpha^{-1}(q_i)$, there exists $\tta^a\in \transs{\tta^a}$ satisfying $(\mathcal{T}^u(x,\cdot),\tta_i^a)\in\csig_\alpha$ and $ \mathcal{T}^u(x,\alpha^{-1}(q_j)) = \tta_{ij}^a$  for all  $j\in \{1,2,\cdots, n\}$. 
\end{itemize}
The converse abstraction is defined in a similar way.
\end{deff}

\begin{rem}
Heuristically,  we stand from the side of the original system and require an abstraction to  
\begin{itemize}
    \item contain states with the same labels as states of the original system;
    \item include transitional measures with the same measuring results on all the discrete states given  any starting point of the original system that can be mapped to an abstract state.
\end{itemize}
Given a rectangular partition and the existence of an abstraction,  one immediate consequence is that the transition matrices are able to recover all possible transition probabilities (of the original system) from a grid to another. 
\end{rem}

\begin{ass}\label{ass: part}
Without loss of generality, we assume that the labelling function is amenable to a rectangular partition\footnote{See e.g.  \cite[Definition 1]{dutreix2020specification}.}. In other words, a state-level abstraction can be obtained from a rectangular partition.
\end{ass}

\section{SOUNDNESS OF ROBUST BMDP ABSTRACTIONS }\label{sec: bmdp}

BMDPs are quasi-controlled Markov systems on a discrete state space with upper/under approximations ($\overt^u$/$\undert^u$) of the real transition matrices. 
\begin{deff}
A BMDP is a tuple $\mathcal{IA}=(\Qs,\act,\{\undert\}, \{\overt\},\ap,L_\mathcal{IA})$, where
\begin{itemize}
    \item $\Qs$ is an $(N+1)$-dimensional state-space for any  $N$, which is obtained by a  finite state-space partition  containing $\{\Delta\}$, i.e., $\Q=(q_1,q_2,\cdots,q_N, q_{N+1}:=\Delta)^T$;
    \item $\act$ is a finite-dimensional actions;
    \item $\ap$ and $L_\mathcal{IA}$ are the same as in Definition \ref{def: MDP};
    \item $\{\undert\}:=\{\undert^a\}_{a\in\act}$ is a family of $N\times N$ matrix such that $\undert_{ij}^a$ is the lower bound of transition probability from the state number $i$ to $j$ for each $i,j\in\{1,2,\cdots,N\}$ and action $a\in\act$ ;
    \item $\{\overt\}:=\{\overt^a\}_{a\in\act}$ is a family of $N\times N$ matrix such that $\overt_{ij}^a$ is the upper bound of transition probability from the state number $i$ to $j$ for each $i,j\in\{1,2,\cdots,N\}$ and action $a\in\act$.
\end{itemize}
\end{deff}

By adding constraints
\begin{equation}
\begin{split}
   [\![\tta^a]\!] & =\{\tta^a: \text{stochastic matrices with}\; \undert^a\leq \tta^a\leq \overt^a\\
   & \qquad\qquad\text{componentwisely}\},
\end{split}
\end{equation}
we are able to transfer an $\mathcal{IA}$ into a controlled Markov system $\IU$ as in Definition \ref{def: abs_control}, whose $\transs{\tta^a}$'s are well defined sets of stochastic matrices for each $a\in\act$. We call the induced $\IU$, which is verified to satisfy Definition \ref{def: abs_control}, the abstraction generated by the BMDP $\mathcal{IA}$, or simply the BMDP abstraction.

\begin{rem}\label{rem: bmdp}
To make $\IU$ an abstraction for \eqref{E: sysr_control}, we can discretize both $\xx$ and $\uu$, such that each node $a\in\act$ represents a grid of $u\in\uu$. We then need the approximation to be such that $\undert_{ij}^a\leq \int_{\alpha^{-1}(q_j)}\mathcal{T}^u(x,dy)\leq \overt_{ij}^a$ for a $u\in a$, for all $x\in \alpha^{-1}(q_i)$  and $i,j=1,\cdots,N$, as well as $\tta_{N+1}=(0,0,\cdots,1)$.
\end{rem}

For any realization of a sequence of actions $a:= \{a_i\}_{i\in\N}$, the controlled Markov system $\IU^a$ is reduced to a family of perturbed  Markov chains  generated by the uncertain choice of $\{\tta\}$ for each $t$. The $n$-step transition  are derived based on $[\![\tta^{a_i}]\!]$:
\begin{equation*}
    \begin{split}
      [\![\tta^{(2)}]\!]&:=\{\tta_0^{a_0}\tta_1^{a_1}:\;\;\tta_i^{a_i} \in[\![\tta^{a_i}]\!], i = 0, 1\},\\
      & \vdots\\
      [\![\tta^{(n)}]\!]&:=\{\tta_0^{a_0}\tta_1^{a_1}\cdots\tta_n^{a_n}:\;\;\tta_i^{a_i}\in[\![\tta^{a_i}]\!],\\
      &\qquad\;i=0,1,\cdots,n\}.\\
    \end{split}
\end{equation*}
The weak compactness and convexity of the probability laws of $\IU^a$ are proved in \cite[Section 3.2]{meng2022robustly}. We also kindly refer readers to the arXiv version \cite[Section 3.1]{meng2022complete} for more details on the weak topology properties. 

Taking the advantages of the above properties, we now show the soundness of BMDP abstractions. 

\begin{deff}
Given a state-level abstraction $\alpha$ and a measure-level abstraction $\csig_\alpha$ from $\MU$ to $\IU$. Let $\phi$ and $\kappa$ be some control policies of $\MU$ and $\IU$, respectively. Recall notations in 
\eqref{E: history_notation}. We call $\phi$ a $\csig_\alpha$-implementation of $\kappa$ if, for each $t\in\N$, 
$$\uf_t=\phi_t(X_{[0,t]},\uf_{[0,t-1]}), \;\;X\in\MU^{\phi}$$
is chosen according to 
$$\af_t=\kappa_{t}(I_{[0,t]},\af_{[0,t-1]}), \;\;I\in\IU^\kappa$$
in a way that, for any realization $u$ and $a$ of $\uf_t$ and $\af_t$,  for any $\mathcal{T}^u\in\transs{\mathcal{T}^u}$ and all $x\in\alpha^{-1}(q_i)$, there exists $\tta^a\in \transs{\tta^a}$ satisfying $(\mathcal{T}^u(x,\cdot),\tta_i^a)\in\csig_\alpha$ and $ \mathcal{T}^u(x,\alpha^{-1}(q_j)) = \tta_{ij}^a$  for all  $j\in \{1,2,\cdots, n\}$. 
\end{deff}
We can define the converse implementation from $\IU$ to $\MU$ based on a converse measure-level relation (from $\IU$ to $\MU$) in a similar way. 
\begin{rem}
Heuristically, a control policy $\kappa$ is generated in the finite-state finite-action abstraction model within $\IU$ to ensure a probabilistic satisfaction of some specification. The selection of the control policy $\phi$ is subjected to $\kappa$ and hence $\IU^\kappa$ according to the abstraction relation, such that (2) of Definition \ref{def: abs_control} is guaranteed. 
\end{rem}

\begin{prop}
Let $\IU$ be a controlled Markov system that is derived from a BMDP with any initial distribution $\mu_0$. Then for any $\omega$-regular specification $\Psi$, given any admissible deterministic control policy $\kappa$, the set $$S^{\mu_0,\kappa}=\{\pp^{\mu_0,\kappa}_I(I^\kappa\vDash\Psi)\}_{I^\kappa\in\IU^\kappa}$$ is a compact interval. 
\end{prop}

\begin{proof}
The proof is similar to \cite[Theorem 2]{meng2022robustly}. We only show the sketch. Let $\af$ be the control input process generated by $\kappa$ such that $\af_t=\kappa_t(I_{[0,t]},\af_{[0,t-1]})$ for each $t$. Note that $\af\in\mathscr{B}(\act^\infty)$ and $\af_t\in \mathscr{B}(\act)$, where the set of actions $\act$ admits a discrete topology. The weak compactness of the probability law $\{\pp_I^{\mu_0,\kappa}\}_{I^\kappa\in\IU^\kappa}$ follows exactly the same reasoning as in  \cite[Proposition 1]{meng2022robustly}. The convexity of every finite-dimensional distribution of $I^\af$ can be obtained in similar way as in \cite[Theorem 2]{meng2022robustly} based on the transition procedure, i.e., for any $q_0,q_{n_1},\cdots,q_{n_t}\in\Q$, 
\begin{equation*}
    \begin{split}
        &\pim_I^{q_0,\kappa}\left[I_0=q_0, \cdots, I_t=q_{n_t}, I_{t+1}=q_{n_{t+1}}\right]\\
       \in&\{\tta_{n_{t+1},n_t}^{a_t}\tta_{n_t,n_{t-1}}^{a_{t-1}}\cdots\tta_{n_1,0}^{a_0}\delta_{q_0}: \tta^{a_i}\in [\![\tta^{a_i}]\!],\\
       & \quad i\in\{0,\cdots,t\},\;\text{and}\;a_t=\kappa(I_{[0,t]}=q_{[0,t]}, a_{[0,t-1]})\}.
    \end{split}
\end{equation*}
By a standard monotone class argument, the convexity for any Borel measurable set $A\in \mathscr{F}$ measured in the set of laws $\pim_I^{q_0,\kappa} $ are guaranteed, which implies the convexity of $S^{q_0, \kappa}$, and hence that of $S^{\mu_0, \kappa}$. 
\end{proof}

The soundness regularity is provided as follows.
\begin{thm}\label{thm: inclusion_control}
Let $\MU$ as in \eqref{E: sysr_control} be a controlled Markov system driven by \eqref{E: sys_abstraction_control}.  Suppose that there exist a state-level abstraction $\alpha$, a measure-level abstraction $\csig_\alpha$, and a BMDP abstraction $\IU$ such that $\MU\preceq _{\csig_\alpha} \IU$. Let $\Psi$ be an $\omega$-regular specification.  Suppose the initial distribution $\nu_0$ of $\MU$ is such that $ \nu_0(\alpha^{-1}(q_0))=1$. Then, given an admissible deterministic control policy $\kappa$, there exists a $\csig_\alpha$-implementation policy $\phi $ of $\kappa$ such that $$ \ppp_X^{\nu_0,\phi}(X^\phi\vDash\Psi)\in \{\pp_I^{q_0,\kappa}(I^\kappa\vDash\Psi)\}_{I^\kappa\in\IU^\kappa},\;\;X^\phi\in\MU^\phi.$$
\end{thm}

\begin{proof}
We denote by $\mu_t$ and $\nu_t$, respectively,  the marginal probability measures on $\mathscr{B}(\Q)$ and $\mathscr{B}(\xx)$ for $t\in\N$. We also use the shorthand notation $\mu_t^a(\cdot) := \mu_t(\;\cdot\;|\af_{t-1}=a)$ and $\nu_t^u(\cdot) := \nu_t(\;\cdot\;|\uf_{t-1}=u)$ to indicate the conditional probabilities. 
We consider $\nu_0=\delta_{x_0}$ a.s. for simplicity. Note that, at $t=1$, by the definition of BMDP abstraction and Remark \ref{rem: bmdp}, there exists a  ${u_0}\in \uu$ such that, 
\begin{equation*}
    \begin{split}
        \undert_{ij}^{a_0} &\leq\nu_1^{u_0}(\alpha^{-1}(q_j)) \\
        & =\int_{\alpha^{-1}(q_j)}\delta_{x_0}\mathcal{T}^{u_0}(x_0,dy) \\
        & \leq \overt_{ij}^{a_0} ,\;\;\forall x_0\in q_0\;\text{and}\;\forall j\in\{1,2,\cdots, N+1\}, 
    \end{split}
\end{equation*}
where $a_0=\kappa_0(q_0)$, and ${u_0}$ is selected accordingly such that the above relation is satisfied. 

We can easily check that $$\mu_1^{a_0}=(\nu_1^{u_0}(\alpha^{-1}(q_1)),\cdots,\nu_1^{u_0}(\alpha^{-1}(q_{N+1})))^T$$ is a proper marginal probability measure of $\IU$ at $t=1$. In particular, $\mu_1^{a_0}(q_j) =\nu_1^{u_0}(\alpha^{-1}(q_j))$ for each $j\in\{1,2,\cdots,N+1\}$.

Similarly, at $t=2$, we have
\begin{equation*}
    \begin{split}
        \undert_{ij}^{a_{1i}}\mu_1^{u_0}(q_i)
        & \leq \nu_2^{u_{1i}}(\alpha^{-1}(q_j))\\       & = \int_{\alpha^{-1}(q_j)}\int_{\alpha^{-1}(q_i)}\nu_1^{u_0}(dx)\mathcal{T}^{u_{1i}}(x,dy)\\
        &\leq \overt_{ij}^{a_{1i}}\mu_1^{u_0}(q_i),\;\forall i,j\in\{1,2,\cdots, N+1\}, 
    \end{split}
\end{equation*}
where $a_{1i}=\kappa_1(q_i)$ for each $i\in\{1,2,\cdots,N+1\}$, and ${u_{1i}}$ is selected accordingly such that the above relation is satisfied. Then, $$\mu_2^{a_{1i}}=(\nu_2^{u_{1i}}(\alpha^{-1}(q_1)),\cdots,\nu_2^{u_{1i}}(\alpha^{-1}(q_{N+1})))^T$$ is again a proper marginal probability measure of $\IU$ at $t=2$ for each $i\in\{1,2,\cdots,N+1\}$.
In addition, there also exists a $\pim^{q_0,\kappa}$ such that its one-dimensional marginals up to $t=2$ admit $\mu_1$ and $\mu_2$, and
satisfies
\begin{equation*}
\begin{split}
        &\pim^{q_0, \kappa}[I_0=q_0, \af_0 = a_0, I_1=q_{i}, \af_1 = a_{1i}, I_2=q_{j}]\\
      =   & \mu_0(q_0)\mu_1^{a_0}(q_i)\mu_2^{a_{1i}}(q_j)\\
      = & \delta_{q_0}(q_0)\int_{\alpha^{-1}(q_i)}\nu_1^{u_0}(dx)\int_{\alpha^{-1}(q_j)}\nu_2^{u_{1i}}(dy)\\
      = & \ppp_X^{q_0,\uf}[X_0=x_0, \uf_0=u_0, X_1\in \alpha^{-1}(q_i),\uf_1=u_{1i},\\
      & \qquad \quad X_2\in\alpha^{-1}(q_j)] 
\end{split}
\end{equation*}
for all $i,j\in\{1,2,\cdots,N+1\}$. We then propagate the process inductively according to the above machinery by 
\begin{itemize}
    \item[1)] selecting $u_t:=u_{tj}$ at each time according to the realization $a_t:=a_{tj}=\kappa_t(q_j)$;
    \item[2)]selecting $\ts^{u_t}$ and $\tta^{a_t}\in\transs{\tta^{a_t}}$ at each time via the connection as the above.
\end{itemize}
We can verify that, by the above selection procedure, there exists  $\pp^{q_0,\kappa}$ such that 
\begin{equation*}
    \begin{split}
        &\pim^{q_0,\kappa}[I_0=q_0, \af_0=a_0, I_1=q_i,\af_1=a_1,\cdots]\\
        =&\ppp_X^{x_0, \uf}[X_0=x_0, \uf_0=u_0, X_1\in \alpha^{-1}(q_i),\uf_2=u_2,\cdots] 
    \end{split}
\end{equation*}
holds for any finite-dimensional distribution. By Kolmogrov extension theorem, there exists a unique probability law $\pim_I^{q_0,\kappa}$ for $(I,\af)$ or $I^\kappa\in\IU^\kappa$ such that it has the same measuring results on any $\mathscr{F}$-measurable sets (recall that $\mathscr{F}=\mathscr{B}(\Q^\infty)$) as the probability law $\ppp_X^{x_0, \uf} $ of the generated process $(X,\uf)$ or $X^\uf$.

The $\csig_\alpha$-implementation $\phi=\{\phi_t\}_{t\in\N}$ exists and is given as $\phi_t(\cdot\;|\;X_{[0,t]},\uf_{[0,t-1]})=\ppp_X^{x_0, \uf}[\uf_t=(\cdot)\;|\;X_{[0,t]},\uf_{[0,t-1]}]$ after averaging out along $\xx^\infty$.
\end{proof}

Based on Theorem \ref{thm: inclusion_control}, we can immediately show whether a control strategy   exists based on the BMDP abstraction such that the controlled process satisfy the probabilistic specification.

\begin{cor}\label{cor: sound_control}
Let $\MU$, its BMDP abstraction $\IU$, an $\omega$-regular formula $\Psi$, and a constant $\rho\in[0,1]$ be given. Suppose there exists a control policy $\kappa$ such that $I^\kappa\vDash\pp^{q_0}_{\bowtie\rho}[\Psi]$ for all $I^\kappa\in\IU^\kappa$, then there exists a policy $\phi$ such that $X^\phi\vDash\ppp^{\nu_0,\phi}_{\bowtie\rho}[\Psi]$ for all $X^\phi\in\MU^\phi$ with $\nu_0(\alpha^{-1}(q_0))=1$. 
\end{cor}

\begin{rem}
The purpose of abstraction-based formal methods is in general different from constructing numerical solutions for SDEs. 
The numerical analysis for SDEs is 
to determine how good the
approximation is and in what sense it is close to the exact solution \cite{kloeden1992stochastic}.

Aside from the analysis based on the time discretization, the stochastic driving forces in discrete-time numerical simulations are given with discrete distributions \textit{in a priori}. For example, a spatial step size should be provided to generate a pseudo random variable from a Gaussian distribution. Consequently, there is a  unique solution in the  discrete canonical space driven by this discrete noise.  The discretized measure of any random variable already provides a deviation from the real measure to begin with. The numerical simulation provides a much smaller set of measurable sample paths, i.e. a natural filtration $\mathscr{F}^{\wb,d}$ subjected to the discrete version of noise $\wb^d$ rather than  $\ff$ (recall Definition \ref{def: clarification}). The missing transitions or measurable sample paths from $\ff$ cannot be recovered given a fixed discretized noise at a time.

On the other hand, from the dual problem point of view, a finite difference approximation for the associated (controlled) Fokker-Plank equation (parabolic equation)
$$\frac{\partial\rho_t}{\partial t}=\LLk^{u,*}\rho_t, $$
where $\LLk^{u,*}$ is the adjoint operator of the infinitesimal generator $\LLk^u$ of \eqref{E: sys_abstraction_control}, provides approximated discrete marginal densities of the probability laws of the solution processes. However, in view of finite-dimensional distribution, this is not sufficient for the evaluation of the probability of sample paths satisfying some linear-time properties over the time horizon. An   approximation should be done for the associated transition semigroups $\{e^{t\LL^{u,*}}\}$ to fulfill such a type of evaluation. 

In comparison with the numerical solutions and the dual approximation of the probability distributions, the stochastic abstractions in this paper do not use the spatially discretized noise as the driving force.  Instead, we directly work on generating a relation based on the state-space discretization such that the transition kernel of the original system is `included' in the discrete family of  transition matrices in the sense of Theorem \ref{thm: inclusion_control}. Even though a refinement of grid size can lead to a convergence for both numerical simulations and stochastic abstractions (see \cite[Proposition 3]{meng2022robustly} for details), they converge from different `directions'. In other words, the family of the discrete probability laws from an abstraction reduces to a singleton whilst the missing transitions in a numerical simulation become empty as the size of the grids converges to $0$. 
\end{rem}

\section{ROBUST COMPLETENESS OF BMDP ABSTRACTIONS}\label{sec: complete}

In this section, we propose the concept of  robustly complete abstractions of discrete-time
nonlinear stochastic systems of the form \eqref{E: sys_abstraction_control} and provide computational procedures
for constructing sound and robustly complete 
abstractions for this class of controlled stochastic systems under mild
conditions.

Note that, in view of the soundness analysis given in Theorem \ref{thm: inclusion_control}, the BMDP abstractions create a formal inclusion of transition probabilities and hence the inclusion of `reachable set' of marginal probability measures. This guarantees that the real satisfaction probability is preserved as in Corollary \ref{cor: sound_control}, however, creates a deviation from the original concrete system. The purpose of completeness analysis in this section is to investigate that, given an arbitrarily small perturbation in a certain sense, whether one can construct a sound BMDP abstraction without providing larger perturbation. To do this, we work on  the space of probability measure metricized by the Wasserstein metric\footnote{This is formally termed as $1^{\text{st}}$-Wasserstein metric. We choose $1^{\text{st}}$-Wasserstein metric due to the convexity and nice property of test functions.} to quantify this extra perturbation.

\subsection{Probability Metrics}

The space of probability measures on a complete, separable, metric (metrisable)  space endowed with the topology of
weak convergence is itself a complete, separable, metric (metrisable) space \cite{billingsley2013convergence}.
 While not easy to compute, the Prohorov metric  can be used to metrize weak topology. We prefer to use Wasserstein metric since it also implies weak convergence and provides more practical meanings in applications. The total variation, on the other hand, implies setwise (conventional) convergence on a continuous base state space $\xx$.

We first recall some basic concepts established in \cite{meng2022robustly} regarding the complete analysis. 

\begin{deff}[Wasserstein distance]\label{def: ws_appendix}
Let $\mu,\nu\in\psp(\xx)$ for $(\xx,|\cdot|)$, the Wasserstein distance is  defined by
$
    \ws{\mu-\nu}=\inf\eee|X-Y|
$, where the infimum is is taken over all joint distributions of the random variables $X$ and $Y$  with marginals $\mu$ and $\nu$ respectively.

We frequently use the following duality form of definition\footnote{$\operatorname{Lip}(h)$ is the Lipschitz constant of $h$ such that $|h(x_2)-h(x_1)|\leq \operatorname{Lip}(h)|x_2-x_1|$.},
\begin{equation*}
\begin{split}
       \ws{\mu-\nu}
       :=&\sup\left\{\left|\int_\xx h(x)d\mu(x)-\int_\xx h(x)d\nu(x)\right|,\right. \\
       &\left.\qquad \; h\in C(\xx),\operatorname{Lip}(h)\leq 1\right\}.  
\end{split}
\end{equation*}
The discrete case, $\ws{\;\cdot\;}^d$, is nothing but to change the integral to summation. Let $\bw=\{\mu\in\psp(\xx): \ws{\mu-\delta_0}< 1\}$. 
Given a set $\G\subseteq\psp(\xx)$, we denote $\|\mu\|_\G=\inf_{\nu\in\G}\ws{\mu-\nu}$  by the distance from $\mu$ to $\G$, and $
    \G+r\bw:=\{\mu:\;\|\mu\|_\G< r\}
$\footnote{This is valid by definition.} by the $r$-neighborhood of $\G$.
\end{deff}

Note that $\bw$ is dual to $\oballr$. For any $\mu\in \bw$, the associated random variable $X$ should satisfy $\eee|X|\leq 1$, and vice versa. 

We also frequently use the following inequalities to bound the  Wasserstein distance between two Gaussians, where the R.H.S. of \eqref{E: ws} is the $2^{\text{nd}}$-Wasserstein distance for two Gaussians.

\begin{prop}\label{prop: compare}
Let $\mu\sim\mathcal{N}(m_1,\Sigma_1)$ and $\nu\sim\mathcal{N}(m_2,\Sigma_2)$ be two Gaussian measures on $\R^n$. Then
\begin{equation}\label{E: ws}
\begin{split}
    |m_1-m_2|& \leq \ws{\mu-\nu}\\
    & \leq \left(\|m_1-m_2\|_2^2+\|\Sigma_1^{1/2}-\Sigma_2^{1/2}\|_F^2\right)^{1/2},
\end{split}
\end{equation}
where $\|\cdot\|_F$ is the Frobenius norm. 
\end{prop}

\begin{deff}[Total variation distance]
Given two probability measures $\mu$ and $\nu$ on $\mathscr{B}(\xx)$, the total variation distance is defined as
\begin{equation}
    \tv{\mu-\nu}=2\sup_{\Gamma\in\mathscr{B}(\xx)}|\mu(\Gamma)-\nu(\Gamma)|.
\end{equation}
    In particular, if $\xx$ is a discrete space, 
    \begin{equation}
         \tv{\mu-\nu}^d=\|\mu-\nu\|_1=\sum_{q\in\xx}|\mu(q)-\nu(q)|.
    \end{equation}
\end{deff}
\begin{rem}
It is equivalent to use the dual representation
\begin{equation}
    \tv{\mu-\nu}=\sup_{\substack{ \|h\|_\infty\leq 1} }\left|\int_\xx h(x)\mu(dx)-\int_\xx h(x)\nu(dx)\right|.
\end{equation}
In this view, total variation distance is not suitable to metrisize weak convergence since it implies a much stronger uniform norm given $\xx$ is continuous. However, working on discrete topology of a finite set, we have the following well known connection.
\begin{equation}\label{E: ws_tv}
    \ws{\mu-\nu}^d = \frac{1}{2}\tv{\mu-\nu}^d.
\end{equation}
This equivalence \cite[Theorem 4]{gibbs2002choosing} on the discrete topology, on the other hand, implies that  abstractions already exist unnecessarily in a functional space with stronger convergence concept. 
\end{rem}

\subsection{Construction of Robustly Complete BMDP Abstractions}

We consider two continuous-state systems with parameters $\ep_2>\ep_1 \ge 0$. The first system, denoted by $\MU_1$, is given by 
\begin{equation}\label{E: control_sys_1}
    X_{t+1}=f(X_t,\uf_t)+b(X_t)\wb_t+\ep_1\xi_t^{(1)},\;\;\xi_t^{(1)}\in\ball,
\end{equation}
and the second system, denoted by $\MU_2$, is driven by 
\begin{equation}\label{E: control_sys_2}
    X_{t+1}=f(X_t,\uf_t)+b(X_t)\wb_t+\ep_2\xi_t^{(2)},\;\;\xi_t^{(2)}\in\ballr.
\end{equation}
We construct a sound and robustly complete BMDP abstraction $\IU$ for $\MU_1$ in a similar way as in \cite{meng2022robustly}, i.e., we build a state-level relation $\alpha$ and a measure-level $\csig_\alpha$ such that 
$$\MU_1\preceq_{\csig_\alpha}\IU,\;\;\IU\preceq_{\csig_\alpha^{-1}} \MU_2.$$

We  define the set of transition probabilities of $\MU_i$, for each fixed $u\in\uu$, from any box $[x]\subseteq\R^n$ as $$\T_i^u([x])=\{\mathcal{T}^u(x,\cdot):\;\mathcal{T}^u\in\transs{\mathcal{T}^u}_i,\;x\in[x]\},\;i=1,2.$$
The following lemma is to straightforward based on \cite[Lemma 3]{meng2022robustly}.

\begin{lem}\label{lem: inclu_control}
Fix any $\ep_1\geq 0$, any box $[x]\subseteq\R^n$, and $u\in\uu$. For all $k>0$, there exists a finitely terminated algorithm to compute an over-approximation of the set of (Gaussian) transition probabilities from $[x]$, 
such that
$$\T_1^u([x])\subseteq \wh{\T_1^u([x])}\subseteq\T_1^u([x])+k\ballr_W, $$
where $\wh{\T_1^u([x])}$ is the computed over-approximation set of Gaussian measures.
\end{lem}

Lemma \ref{lem: inclu_control} is to construct an over-approximation $\wh{\T_1^u([x])}$ of the set of Gaussian transition probabilities from the original concrete system $\MU_1$, such that any Gaussian measure within $\wh{\T_1^u([x])}$  will not perturb the mean more than any arbitrarily small $k$. We skip the proof due to the similarity to \cite[Lemma 3]{meng2022robustly} and \cite[Lemma 1]{liu2017robust}. The main step is to find inclusion functions for $f$ and $b$, as well as a mesh of $[x]$ with appropriate size. The over-approximation of the `reachable' mean and covariance can be obtained by union the regions generated by inclusion functions acting on the mesh. 

Note that, recalling  Definition \ref{def: clarification2},  we are actually working on the quantification for the stopped processes. 
The introduce a modification that does not affect the law of the stopped processes, i.e., we use a weighted point mass to represent the measures at the boundary, and the mean value should remain the same.

\begin{deff}\label{def: mod}
For $i=1,2$, we introduce the modified transition probabilities for $\MU_i=(\xx,\uu,  \{\mathcal{T}\}_i, 
    \ap,L_\MU)$. For any $u\in\uu$, for all $\mathcal{T}_i^u\in\transs{\mathcal{T}}_i^u$, let \begin{equation}\label{E: mod}
    \ttil_i^u(x,\Gamma)=\left\{\begin{array}{lr} \mathcal{T}_i^u(x,\Gamma),\; \forall\Gamma\subseteq\Ws,\;\forall x\in \Ws,\\
\mathcal{T}_i^u(x,\Ws^c),\;\Gamma=\partial\Ws,\;\forall x\in \Ws,\\
1,\;\Gamma=\partial\Ws,\;x\in\partial\Ws.
\end{array}\right.
\end{equation}
Correspondingly, let $\widetilde{\transs{\mathcal{T}^u}}$ denote the collection. Likewise, we also use $\widetilde{(\;\cdot\;)^u}$ to denote the induced quantities of any other types w.r.t. such a modification. 
\end{deff}

We are now ready to show the existence of a robustly complete abstraction given \eqref{E: control_sys_1} and \eqref{E: control_sys_2}. 

\begin{thm}\label{thm: main}
For any $0\leq \ep_1<\ep_2$, we consider $\MU_i=(\xx, \uu,  \{\widetilde{\mathcal{T}}\}_i, \ap,L_\MU)$, $i=1,2$, that are driven by  \eqref{E: control_sys_1} and \eqref{E: control_sys_2}, respectively. 
Then, under Assumption \ref{ass: part}, there exists a rectangular partition $\Q$ (state-level relation $\alpha\subseteq\xx \times \Q$), 
 a measure-level relation $\csig_\alpha$ 
and 
a finite-state abstraction system $\I=(\Q,\act, \{\tta\},\ap,L_\IU)$ such that
\begin{equation}
    \MU_1\preceq _{\csig_\alpha}\IU,\;\;\IU\preceq_{\csig_\alpha^{-1}}  \MU_2.
\end{equation}
\end{thm}

\begin{proof}
We construct a finite-state BMDP abstraction in a similar way as in \cite[Theorem 4]{meng2022robustly}. Aside from the additional dependence on the control inputs, we also  provide tighter estimations on sets of probability measures. By Assumption \ref{ass: part}, we use uniform rectangular partition $\Q$ on $\Ws$. We then let the state-level relation be $\alpha=\{(x,q): q=\eta\lfloor\frac{x}{\eta}\rfloor\}\cup\{(\Delta,\Delta)\}$, and $\act=\{a: \varrho\fl{\frac{u}{\varrho}}\}$,  where $\lfloor\cdot\rfloor$ is the floor function. The parameters $\eta$, $\varrho$ are to be chosen later. Denote the number of discrete nodes by $N+1$. 
    
     We construct the measure-level abstraction as follows. We repeat the procedure with updated notations for the control systems.   For any fixed $u=a\in\act$,  for any $\tdt^u\in\tmou_1$ and $q\in \Q$,
 \begin{enumerate}
        \item[1)] 
        for all $\tnu^u\sim\widetilde{\mathcal{N}}(m,s^2)\in\widetilde{\T}_1^u(\alpha^{-1}(q),\cdot)$, store $\{(m_{l},s_{l})=(\eta\lfloor\frac{m}{\eta}\rfloor,\eta^2\fl{\frac{s^2}{\eta^2}})\}_l$;
        \item[2)] for each $l$, define $\tnu_{l}^{u,\rff}\sim\widetilde{\mathcal{N}}(m_{l},s_{l})$ (implicitly, we need to compute $\nu_{l}^{u,\rff}(\alpha^{-1}(\Delta))$); compute $\tnu_{l}^{u,\rff}(\alpha^{-1}(q_j))$ for each $q_j\in \Q\setminus\Delta$; 
        \item[3)] for each  $l$,  define $\mu^{u,\rff}_{l}=[\tnu_{l}^{u,\rff}(\alpha^{-1}(q_1)),\cdots,\tnu_l^{u,\rff}(\alpha^{-1}(q_N)),\tnu_l^{u,\rff}(\alpha^{-1}(\Delta)]$;
        \item[4)] compute
       $ \textbf{ws}:=(\sqrt{2n}+2)\eta 
        $ and $\textbf{tv}:=2\cdot \textbf{ws} $;
        \item[5)] construct $\transs{\mu^u}=\bigcup_l\{\mu:\tv{\mu-\mu^{u,\text{ref}}_{l}}\leq \textbf{tv}(\eta),\; \mu(\Delta)+\sum_j^N\mu(q_j)=1 \}$;  
        \item[6)]let $\csig_\alpha:=\{(\tnu^u,\mu^u),\;\mu^u\in\transs{\mu^u}\}$ be a relation between $\tnu^u\in\widetilde{\T}^u(\alpha^{-1}(q))$ and the generated $\transs{\mu^u}$.
    \end{enumerate}   
 Repeat the above step for all $q$ and then for all $u=a\in\act$, the relation $\csig_\alpha$ is obtained.  We  denote $\G_i^u:=\widetilde{\T}^u_1(\alpha^{-1}(q_i),\cdot)$ and $\widehat{\G}^u_i:=\wh{\widetilde{\T}^u_1(\alpha^{-1}(q_i),\cdot)}$.
 
 \noindent
    \textbf{Step 1}:
 For each $u=a\in\act$, for $i\leq N$, let $\transs{\tta_i^u}= \csig_\alpha(\widehat{\G}_i^u)$ and the transition  collection be $\transs{\tta^u}$. It can be shown  that the finite-state BMDP $\IU$   abstracts $\MU_1$ based on Definition \ref{def: abs_control}:  for each $a$, there exists $u\in\uu$ (where we set it to be $a$), such that for any $\tilde{\nu}^u\in\G_i^u$ and hence in $\widehat{\G}_i^u$, there exists a discrete measures in $\tta_i^u\in\csig_\alpha(\widehat{\G}_i^u)$  such that for all $q_j$ we have  $\tilde{\nu}^u(\alpha^{-1}(q_j))=\tta_{ij}^u$. 
 
 The proof is done by the exact same way as the proof of  \cite[Claim 1, Theorem 4]{meng2022robustly} for each fixed control input. We summarize the methodology as follows:
 \begin{itemize}
     \item[a)] To not miss any possible transition of $\MU_1$ from each $x\in\alpha^{-1}(q_i)$, we work on the over-approximation set $\widehat{\G}_i^u$ of Gaussian measures. It can be easily verified that, within the abstractions, we also have $\Sigma_\alpha(\wh{\G}_i^u)\supseteq \Sigma_\alpha(\G_i^u)$. Now we verify that $\Sigma_\alpha$ given in 6) is indeed a valid relation that creates an abstraction. 
     \item[b)] For any modified Gaussian $\tilde{\nu}^u\in\hat{\G}_i^u$, there  exists a $\tilde{\nu}^{u,\rff}$ such that the distance is bounded:  $\ws{\tilde{\nu}-\tilde{\nu}^{u,\rff}}\leq\ws{\nu-\nu^{u,\rff}}\leq \sqrt{2n}\eta$. This is estimated by Proposition \ref{prop: compare}.
     \item[c)] Reflecting on the space of discrete measures (a row of an abstraction matrix), we have the following inflation
     \begin{equation}
        \begin{split}
         &\ws{\mu-\mu^{u,\rff}}^d\\ 
          \leq& \ws{\mu-\tnu}+\ws{\tnu-\tnu^{u,\rff}}+\ws{\tnu^{u,\rff}-\mu^{u,\rff}}\\
                  \leq& \; \textbf{ws},
        \end{split}
    \end{equation}
 \end{itemize}
 where the first  and third term above is to connect a discretized measure with a continuous measure. Note that for any continuous measure $\mathfrak{v}$ and its discretized version $\mathfrak{m}$, we have
\begin{equation}\label{E: bound}
        \begin{split}
            &\ws{\mathfrak{m}-\mathfrak{v}}\\
            =&\sup_{h\in C(\xx),\operatorname{Lip}(h)\leq 1}\left|\int_\xx h(x)d\mathfrak{m}(x)-\int_\xx h(x)d\mathfrak{v}(x)\right|\\
            \leq &\sup_{h\in C(\xx),\operatorname{Lip}(h)\leq 1}\sum_{j=1}^n\int_{\alpha^{-1}(q_j)}|h(x)-h(q_j)|d\mathfrak{v}(x)\\
            \leq& \eta \sum_{j=1}^n\int_{\alpha^{-1}(q_j)}d\mathfrak{v}(x)\leq \eta.
        \end{split}
    \end{equation}
 By 5) an 6) and  Definition \ref{def: abs_control}, we have stored such $\mu$ centered at the reference measure w.r.t. the total variation distance, and this collection has sufficient amount of transition matrices as a valid abstraction by definition. 
 
 \noindent\textbf{Step 2:} 
 Now we choose of $\eta$ and $\varrho$ such that the constructed BMDP abstraction can be abstracted by $\MU_2$ via the converse relation $\csig_\alpha^{-1}$. Note that $a\in u+\varrho\ball$ for any $u\in\uu$.  We need to choose $\eta$, $\varrho$ and $k$ sufficiently small such that
\begin{equation}
    2\eta+1/2\cdot\textbf{tv}(\eta)+L\rho+k\leq \ep_2-\ep_1,
\end{equation}
where $L$  is the Lipschitz constant of $f$. Then, 
 we have 
\begin{equation}\label{E: converse_control}
    \begin{split}
        \csig_\alpha^{-1}(\csig_\alpha(\widehat{\G}_i^u))&\subseteq \widehat{\G}_i^u+(2\eta+1/2\cdot\textbf{tv}(\eta))\cdot\ballr_W+L\varrho\ball\\
        &\subseteq \G_i^u+(2\eta+1/2\cdot\textbf{tv}(\eta)+L\varrho+k)\cdot\ballr_W 
    \end{split}
\end{equation}
for each $i$. 
Note that  all the `ref' information is recorded, and, particularly, for any $\mu^u\in\Sigma_\alpha(\G_i^u)$ there exists a $\mu^{u,\rff}$ within a total variation radius $\textbf{tv}(\eta)$. 

The inclusions in \eqref{E: converse_control} are  to conversely find all possible corresponding  measure $\tnu^u$ 
            that matches $\mu^u$ by their probabilities on discrete nodes. 
            All such $\tnu^u$ should satisfy,
            \begin{equation}\label{E: incl}
                \begin{split}
                    &\ws{\tnu^u-\tnu^{u,\rff}}\\
                    \leq  & \ws{\tnu-\mu}+\ws{\mu-\mu^{u,\rff}}^d+\ws{\mu^{u,\rff}-\tnu^{u,\rff}}\\
                  \leq & 2\eta+1/2\cdot\textbf{tv}(\eta),
                \end{split}
            \end{equation}
    where the bounds for the first and third terms are obtained in the same way as \eqref{E: bound}.
The second term is improved compared to \cite{meng2022robustly} based on the connection \eqref{E: ws_tv}. 

By the construction, we can verify that for each $u\in\uu$, there exists an $a\in\act$ (guaranteed by the finite covering relation $a\in u+\varrho\ball$) such that  the choice in \eqref{E: converse_control} makes $\csig_\alpha^{-1}(\csig_\alpha(\widehat{\G}_i^u))\subseteq\widetilde{\T}_2^u(\alpha^{-1}(q_i))$, which completes the proof.
\end{proof}

\begin{rem}
As noted in \cite{meng2022robustly}, the key point of the construction in Theorem \ref{thm: main} is to record the `ref' points and corresponding radii, which form finite coverings of the compact space of measures. We use `finite-state' instead of `finite' abstraction because we do not further discretize the dual space of the solution processes, which is the space of probability measures. 
\end{rem} 

\begin{rem}
As shown in Step 1.b), we estimate the Wasserstein distance with the reference measure by the second moment difference of the associated random variables. In this view, we can also replace the additional uncertainty  $\xi^{(1)}$, which is a sequence of   point-mass perturbations, by a sequence of $\ls^1$ independent noise with known, bounded second moment. In this case, we can still eventually obtain a robust complete abstraction by a similar methodology. 
\end{rem}

Theorem \ref{thm: main} shows that, given any $0\leq \ep_1<\ep_2$, there exists a sufficiently and  necessarily refined uniform discretization of $\xx$, as well as a measure-level relation such that a robustly complete abstraction $\IU$ can be constructed. 
We can algorithmically
synthesize a control strategy for $\MU_1$ by generating 
$\IU$ and then solving a discrete synthesis problem for $\IU$
with some probabilistic specification. In view of Corollary \ref{cor: sound_control}, if a control strategy $\kappa$ exists to fulfill the probabilistic specification for $\IU^\kappa$, then there exists a policy $\phi$ to guarantee the satisfaction of $\MU_1^\phi$. On the other hand, if there is no policies to realize a specification for $\IU$, then the system $\MU_2$ is also not realizable w.r.t. the same specification. The latter is implied by the following corollary. 

\begin{cor}\label{cor: sound_control_converse}
Given a  specification formula $\Psi$, let $S_2^{\nu_0,\phi}=\{\ppp_X^{\nu_0,\phi}(X\vDash\Psi)\}_{X^\phi\in\MU_2}$ be the set of the satisfaction probability of $\Psi$ under a control policy $\phi$ for the system $\MU_2$. Then, for each control policy $\phi$ of $\MU_2$, there exists a policy $\kappa$ for $\IU$ such that $ S_\IU^{q_0,\kappa}\subseteq S_2^{\nu_0,\phi}$ for any initial conditions satisfying $\nu_0(\alpha^{-1}(q_0))=1$, where $S_\IU^{q_0,\kappa}=\{\pim_I^{q_0,\kappa}(I\vDash\Psi)\}_{I^\kappa\in\IU}$. Both $S_\IU^{q_0,\kappa}$ and $S_2^{\nu_0,\phi}$ are compact. 
\end{cor}

\begin{proof}
The  inclusion and compactness (for each policy) is done in a similar way as Theorem \ref{thm: inclusion_control} by the inductive construction of probability laws. 
\end{proof}

\section{A DISCUSSION ON  STOCHASTIC CONTROL SYSTEMS WITH NOISY OBSERVATION}\label{sec: noisy}
In this section, we discuss the case when the observations of the sample paths are corrupted by noise. Since there is no direct access to the exact sample path information, we aim to obtain optimal estimates of the sample path signal based on noisy observations, which is known as the optimal filter. Apart from the nonlinear filtering, the philosophy of constructing sound and robustly complete abstractions for such systems maintain the same. We hence do not reiterate the procedure in this section but rather deliver a discussion on the mathematical complexity of the  potential abstractions. Before we proceed, we briefly introduce the theory of nonlinear filtering. 

\subsection{Nonlinear Filtering for Discrete-Time Systems}
Consider the discrete-time signal and observation of the following form
\begin{subequations}
\begin{align}
&X_{t+1}=f(X_t, \uf_t)+b(X_t)\wb_t,\label{E: prior}\\
&Y_t=h(X_t)+\beta_t, \label{E: observation}
\end{align}
\end{subequations}
where $Y$ is a $\mathcal{Y}$-valued observation via a continuous Borel measurable function $h$ and i.i.d. Gaussian process $\beta:=\{\beta_t\}_{t\in\N}$ with proper dimensions. We also set $\wb$ and $\beta$ to be mutually independent. 

Similar to \eqref{E: history_notation}, for any fixed $t>0$, we define the short hand notation for the history of observation
\begin{equation}
    Y_{[0,t]}:=\{Y_s\}_{s\in[0,t]}
\end{equation}
Unlike the system without corrupted observations, it is natural
to suppose that the selection of a control at time $t$ is based on 
$Y_{[0,t]}$ and 
$\uf_{[0,t-1]}$. 
An admissible control policy $\kappa=\{\kappa_t\}_{t\in\N}$ in this case is such that,  for each fixed $t>0$, we have, for any $\ck\in\mathscr{B}(\uu)$,
\begin{equation}\label{E: policy_partial}
    \kappa_t(\ck\;|\;Y_{[0,t]}; \uf_{[0,t-1]})=\ppp[\uf_t\in\ck\;|\;Y_{[0,t]}; \uf_{[0,t-1]}].
\end{equation}
A deterministic admissible policy $\kappa$ is such that $\uf_t=\kappa_t(Y_{[0,t]}; \uf_{[0,t-1]})$. 

Let $H(Y_t\in A\;|\;X_t=x_t)$, $A\in\mathscr{B}(\mathcal{Y})$, be the observation channel, which is the transition kernel generated by \eqref{E: observation}. Given any initial distribution $\mu_0$ of $X$,  the probability law $\ppp^{\mu_0,\kappa}$ of $(X,Y,\uf):=\{X_t,Y_t,\uf_t\}_{t\in\N}$ can be uniquely determined based on the the transition kernel, observation channel, and the control policy. 

Given a policy $\kappa$ (we set it to be deterministic without loss of generality), the estimation of $X_t$ given $Y_{[0,t]}$ that 
minimizes the mean square error loss is given as $$\Pi_t(\Gamma):=\ppp^{\mu_0,\kappa}[X_t\in\Gamma\;|\;Y_{[0,t]},\uf_{[0,t-1]}], \;\;\Gamma\in\mathscr{B}(\xx).$$ We call this random measure $\Pi_t\in\psp(\xx)$ for each $t$ the optimal filter. Using Bayes rule, we have
\begin{equation}\label{E: brief}
    \begin{split}
        &\Pi_t(\Gamma)\\
        &\quad=\ppp^{\mu_0,\kappa}[X_t\in\Gamma\;|\;Y_{[0,t]},\uf_{[0,t-1]}]\\
        &\quad=\frac{\int_\xx H(Y_{t}|X_{t}=x_{t})\tta_{t-1}^{\uf_{t-1}}(x_{t-1}, \Gamma)\cdot\Pi_{t-1}(dx_{t-1})}{\int_\xx\int_\xx H(Y_{t}|X_{t}=x_{t})\tta_{t-1}^{\uf_{t-1}}(x_{t-1}, dx_t)\cdot\Pi_{t-1}(dx_{t-1})}\\
       &\quad=:F(\Pi_{t-1},Y_{t-1},\uf_{t-1})(\Gamma),
    \end{split}
\end{equation}
where $\uf_{t-1}$ is determined by $\kappa_{t-1}(Y_{[0,t-1]},\uf_{[0,t-2]})$. 
It can also be shown that the process $(\Pi,\uf):=\{\Pi_t,\uf_t\}_{t\in\N}$ is a controlled Markov process \cite{saldi2017finite} with transition probability 
\begin{equation}
    \begin{split}
     & \ppp[\Pi_{t+1}\in D\;|\;\Pi_t=\pi_t,\uf_t=u_t]\\
     = & \int_\mathcal{Y}\mathds{1}_{\{F(\pi_t,y_t,u_t)\in D\}}\cdot \mathfrak{n}(dy_t),\;\;D\in\mathscr{B}(\psp(\xx)).   
    \end{split}
\end{equation}
We also use $\Pi^\uf$ to emphasize the marginal behavior of the process $(\Pi,\uf)$. Given the observations and the adaptively generated control signal,  the optimal estimation of the conditional probability of satisfying any $\omega$-regular formula $\Psi$ is given by 
\begin{equation}\label{E: random_measure}
    \ppp_\Pi^{\mu_0,\uf}[X\vDash\Psi]:=\ppp^{\mu_0,\uf}[X\vDash\Psi\;|\;Y]=\int_{\xx^\infty} \mathds{1}_{\{X\vDash\Psi\}}\;\Pi^\uf(dx). 
\end{equation}

Note that it is difficult to obtain the full knowledge of $Y$, our goal is to generate control policies such that 
the optimal estimation $\ppp^{\mu_0,\uf}[X\vDash\Psi\;|\;Y]$ possesses certain confidence of satisfying the probabilistic requirement given any realization of observation. The above derivation converts the problem into a fully observed controlled Markov process $(\Pi,\uf)$ via an enlargement of the state space, where  control policies and even optimal control policies can be synthesized accordingly for the (hypothetically) fully observed $\Pi$ \cite{saldi2017finite}.  The policy fulfilling the goal mentioned above is thereby decidable. 

The construction  of the optimal  filter process (or the function $F$ in \eqref{E: brief}) can be decomposed into a two-step recursion based on the transition relation in \eqref{E: brief}. \\
\noindent\textbf{Prediction (Prior)}: At time $t$, $\Pi_{t-1}(dx)$ is feed into the r.h.s. of the prior knowledge of the dynamics for $X$, i.e., \eqref{E: prior}. The prediction of $X_t$ based on $Y_{[0,t-1]}$ as well as the $u$ determined at $t$ is such that 
\begin{equation}\label{E: prediction}
    \hat{\Pi}_t(dx)=\int_\xx\tta_t^u(\tx,dx)\Pi_{t-1}(d\tx).
\end{equation}
\noindent\textbf{Filtering (Posterior)}: 
This step is to assimilate the observation at the instant $t$, which is given as 
\begin{equation}\label{E: filtering}
    \Pi_t(dx)=\mathfrak{n}(Y_t) H(Y_t|X_t=x)\hat{\Pi}_t(dx),
\end{equation}
where $\mathfrak{n}(Y_t)=\int_\xx H(Y_t|X_t=x)\hat{\Pi}_t(dx)$ is the normalizer. 

For numerical approximation, we simulate and propagate the optimal filter process using matrix approximations of each step's  transition kernel, whereas for formal abstractions, we need to find the `inclusion' of the transitions for each step as usual. 

\subsection{A Brief Discussion on Stochastic Abstractions for Control Systems with Noisy Observations}
Motivated by generating optimal control policies using the  knowledge of  filter process $(\Pi, \uf)$, the stochastic abstractions for partially observed processes can be reduced to obtain a sound and robustly complete abstraction for the process $(\Pi, \uf)$. To convey the idea, we simply consider the following two systems with noisy observations
\begin{subequations}\label{E: filter_1}
\begin{align}
&X_{t+1}=f(X_t, \uf_t)+b(X_t)\wb_t+\ep_1\xi_t^{(1)},\\
&Y_t=X_t+\beta_t+\varsigma_1\zeta_t^{(1)}, 
\end{align}
\end{subequations}
and
\begin{subequations}\label{E: filter_2}
\begin{align}
&X_{t+1}=f(X_t, \uf_t)+b(X_t)\wb_t+\ep_2\xi_t^{(2)},\\
&Y_t=X_t+\beta_t+\varsigma_2\zeta_t^{(2)}, 
\end{align}
\end{subequations}
where $\zeta_t^{(i)}\in\ballr$ are i.i.d. for each $t$ and each $i\in\{1,2\}$, the intensities satisfy $0\leq \varsigma_1<\varsigma_2$. The rest of the notations are as previously mentioned. 

We convert the filter processes that are generated by \eqref{E: filter_1} and \eqref{E: filter_2} into the expression of controlled Markov systems
\begin{equation}
    \mathbb{FU}_i=(\xx,\mathcal{Y}, \uu, \{T\}_i, \transs{H}_i, \ap, L_\mathbb{FU}),\;\;i=1,2,
\end{equation}
where the additional $\mathcal{Y}$ and its collection of observation channel $\transs{H}_i$ are  needed in the filtering step for generating the controlled filter process $\Pi^\uf$. The other notions are the same as previously mentioned.

To find an abstraction for $\mathbb{FU}_1$, we need a state-level relation or discretization $\alpha$ as usual. Then, we need both $\{T\}_1$ and $\transs{H}_1$ to be abstracted via some measure-level relation, so that the transition probability of $\Pi$ is abstracted by a set of discrete transition probabilities given the same set of discrete observations in the sense of (2) of Definition \ref{def: abs_control}.

Now we denote the BMDP abstraction for $\mathbb{FU}_1$ as 
\begin{equation}
    \IU=(\Q, \mathcal{Y}_\Q, \act, \{\tta\}, \transs{H_\Q}, \ap, L_\IU),
\end{equation}
where $\mathcal{Y}_\Q$ is the discretized observation states that are obtained by the state-level relation $\alpha$, and $\transs{H_\Q}$ is the collection of the discrete observation channels that are obtained based on some measure-level relation $\csig_\alpha$. The intuition   of $\IU$ is that we need `more' transitions in the abstraction for the prior knowledge of the dynamics that are related via the measurablility of labelled nodes, as well as `more' transitions for the filtering step to obtain enough observations for decision making. 

\begin{rem}
Note that the collection $\{\tta\}$ for each $u$ can be obtained in the same way as the case without noisy observation. To obtain $\transs{H_\Q}$, we notice that
$$H(dy\;|\; x)=\frac{1}{\sqrt{2\pi}}\exp\left[\frac{(y-x)^2}{2}\right]dy. $$
The over/under-approximation for any $x\in\alpha^{-1}(q_i)$ to the observation $\alpha^{-1}(q_j)$ can be evaluated accordingly. 
\end{rem}

The soundness of $\IU$ for the controlled filter process system $\mathbb{FU}_1$ is in the following sense: given any initial distribution, for each $\kappa$ (based on the discrete observation $Y_\Qs$) for $\IU$, there exists a control policy $\phi$ such that for each $\Pi^\phi\in\mathbb{FU}_1^\phi$, \begin{itemize}
    \item there exists a $\Pi^{d,\kappa}\in\mathbb{IA}^{\kappa}$ whose observation process $Y_\Q$ has the same probability  with $Y$ of $\Pi^\phi$ on each discrete node $q\in\Q$, and $\Pi^{d,\kappa}$ has the same evaluation on all the discrete measurable sets $A\in\mathscr{F}$ with $\Pi^\phi$;
    \item the discrete probability law $\mathbb{P}^{d,\kappa}$ for $\Pi^{d,\kappa}\in\mathbb{IA}^{\kappa}$ forms a convex and weakly compact set;
    \item the optimal estimation satisfies, for a given $p\in[0,1]$, 
    \begin{equation}
    \begin{split}
     &\int_{\Pi^\phi\in\mathscr{B}(\psp(\xx))}\mathds{1}_{\left\{\ppp_\Pi^{\nu_0,\phi}[X^\phi\vDash\Psi]\bowtie p\right\}}\mathbb{P}^{\phi}(d\Pi^\phi)\\
     \in & \left\{\int_{\Pi^{d,\kappa}\in\mathscr{B}(\psp(\Q))}\mathds{1}_{\left\{\pim_\Pi^{\mu_0,\kappa}[X^\phi\vDash\Psi]\bowtie p\right\}}\mathbb{P}^{d,\kappa}(d\Pi^{d,\kappa})\right\}.
    \end{split}
    \end{equation}
\end{itemize}
A proper task is to find an control policy such that the optimal estimation of the probabilistic specification of $X\vDash\Psi$ has a confidence at least $q\in[0,1]$, i.e., $\mathbb{P}^\phi\left(\ppp^{\mu_0,\phi}[X\vDash\Psi\;|\;Y]\bowtie p\right)\geq q$. Then we can search control policies $\kappa$ in $\IU$ for all the filter process $\Pi^{d}$, such that strategy can make the lower bound of  $$\left\{\int_{\Pi^{d,\kappa}\in\mathscr{B}(\psp(\Q))}\mathds{1}_{\left\{\pim_\Pi^{\mu_0,\kappa}[X^\phi\vDash\Psi]\bowtie p\right\}}\mathbb{P}^{d,\kappa}(d\Pi^{d,\kappa})\right\}_{\Pi^{d,\kappa}\in\IU^\kappa} $$
greater than or equal to $q$.

The robust completeness can be verified in a similar way as Section \ref{sec: complete}, except now we need to decompose the procedure to guarantee the robust completeness for both prediction and filtering steps. The discretization need to rely on the value of $\ep_2-\ep_1$ and $\varsigma_2-\varsigma_1$.

Recall Section \ref{sec: bmdp}, where we have compared the abstraction with the numerical simulation of the probability measure using  finite-difference schemes for Fokker-Planck equations. The counterpart of Fokker-Planck equations for evaluating the probability law of the optimal filter in  systems with  noisy observations is the famous Zakai's stochastic partial differential equation\footnote{We omit the content here and kindly refer readers to \cite{budhiraja2007survey} for details.}. The approximation of such a solution   already suffers from the curse of dimensionality. Using formal abstractions to enlarge the partially observed processes into the filter processes with full observations, based on which  control policies can be determined and  utilized back to the partially observed cases, seems tedious and impractical. Besides the theoretical formal guarantee of a confidence of a satisfaction probability (i.e., a probabilistic requirement of the probabilistic specification), the abstraction  essentially solves the continuous probability law of a continuous conditional expectation (or a random measure) upon some process with discrete labels using discrete inclusions. We hence do not recommend readers to complicate the problem.

\section{CONCLUSION}\label{sec: conclusion}

In this paper, we investigated the mathematical properties of  formal abstractions for discrete-time  controlled nonlinear stochastic systems. We discussed the motivation  of constructing sound and complete formal stochastic abstractions  and the philosophy in comparison to numerical approximations in Section \ref{sec: bmdp}. A brief discussion on the extension of stochastic abstractions for controlled stochastic systems with noisy observation was provided in Section \ref{sec: noisy}. The construction of such abstractions can be analogous to solving a discretized version of Zakai's equation via formal inclusions, which suffers from a curse of  excessive dimensionality. 

Our work provides an appropriate mathematical language to discuss the soundness and approximate completeness of abstractions for stochastic systems using BMDP. We show that abstractions with extra uncertainties are not straightforward extensions of their non-stochastic counterparts \cite{liu2017robust, liu2021closing}, and view this as the most significant contribution of our work.

For future work, it would be interesting to design algorithms to construct robustly complete BMDP  abstractions for more general robust stochastic systems with $\ls^1$ perturbations based on the weak topology. The size of state discretization can be refined given more specific assumptions on  system dynamics and linear-time objectives. It is also of a theoretical interest to construct robustly complete abstractions for continuous-time stochastic system and demonstrate the controllability given mild conditions.  Even though we aimed to provide a theoretical foundation of BMDP abstractions for  continuous-state stochastic systems, we hope the results can shed some light on designing more powerful robust control synthesis algorithms.

\bibliographystyle{ieeetr}  
\bibliography{ref}

\vspace*{5pt}

\end{document}